\documentclass[10pt,journal,compsoc]{IEEEtran}

\usepackage{amsmath}
\usepackage[dvips]{color}
\usepackage{comment}
\usepackage{todonotes}
\usepackage{epsf}
\usepackage{epsfig}
\usepackage{times}
\usepackage{epsfig}
\usepackage{graphicx}
\usepackage{bbold}
\usepackage{mathrsfs}
\usepackage{amssymb}
\usepackage{pdfpages}
\usepackage{epstopdf}

\usepackage{balance}
\usepackage{colortbl}
\usepackage{tabulary}
\usepackage{booktabs}

\usepackage{dsfont}
\usepackage{lettrine} % \lettrine[findent=1pt]{{{R}}}{}
\usepackage{amsmath,epsfig,amssymb,algorithm,algpseudocode,amsthm,cite,url}
\usepackage{subcaption}
\allowdisplaybreaks
\usepackage{csquotes}
\usepackage{verbatim}
\usepackage[english]{babel}
\usepackage{amsmath,amssymb}

\usepackage{verbatim}

\newtheorem{theorem}{\bf Theorem}

\newtheorem{proposition}{\bf Proposition}

\newtheorem{definition}{\bf Definition}

\begin{document}
%\pagenumbering{gobble}% Remove page numbers (and reset to 1)
\clearpage
%\maketitle
%Title.

\title{Social Community-Aware Content Placement in Wireless Device-to-Device Communication Networks}

\author{Mehdi~Naderi~Soorki, Walid~Saad, Mohammad~Hossein~Manshaei,  and~Hossein~Saidi
\IEEEcompsocitemizethanks{\IEEEcompsocthanksitem {M. Naderi Soorki, M.H. Manshaei and H. Saidi are with the Department of Electrical and Computer Engineering, Isfahan University of Technology, Isfahan 84156-83111, Iran. E-mails:  m.naderisoorki@ec.iut.ac.ir; \{manshaei, hsaidi\}@cc.iut.ac.ir}

\IEEEcompsocthanksitem {M. Naderi Soorki and W. Saad are with Wireless@VT, Bradley Department of Electrical and Computer Engineering, Virginia Tech, Blacksburg, VA, USA, ‎ E-mail: \{mehdin,walids\}@vt.edu}
}

\thanks{
This research was supported by the U.S. National Science Foundation under Grant CNS-1513697.
}

}
\IEEEtitleabstractindextext{%
\begin{abstract}
In this paper, a novel framework for optimizing the caching of popular user content at the level of wireless user equipments (UEs) is proposed. The goal is to improve content offloading over wireless device-to-device (D2D) communication links. In the considered network, users belong to different social communities while their UEs form a single multi-hop D2D network. The proposed framework allows to exploit the multi-community social context of users for improving  the local offloading of cached content in a multi-hop D2D network. To model the collaborative effect of a set of UEs on content offloading, a cooperative game between the UEs is formulated. For this game, it is shown that the Shapley value (SV) of each UE effectively captures the impact of this UE on the overall content offloading process. To capture the presence of multiple social communities that connect the UEs, a hypergraph model is proposed. Two line graphs, an influence-weighted graph, and a connectivity-weighted graph, are developed for analyzing the proposed hypergaph model. Using the developed line graphs along with the SV of the cooperative game, a precise offloading power metric is derived for each UE within a multi-community, multi-hop D2D network. Then, UEs with high offloading power are chosen as the optimal locations for caching the popular content. Simulation results show that, on the average, the proposed cache placement framework achieves $12\%$, $19\%$, and $21\%$ improvements in terms of the number of UEs that received offloaded popular content compared to the schemes based on betweenness, degree, and closeness centrality, respectively.
\vspace{-.3cm}
\end{abstract}
\begin{IEEEkeywords}
Cache placement, multi-hop D2D network, hypergraph model, Shapley value.
%\vspace{-1cm}
\end{IEEEkeywords}}
\maketitle
\section{Introduction}
\label{sec:Intro}
Due to the massive growth of smart devices and the increasing popularity of bandwidth-intensive applications, mobile traffic is expected to grow continuously at a rapid rate in the next few years~\cite{wang2015social}. Local area services and social network services are expected to constitute a major portion of this mobile traffic. In both local area and mobile social networking services, a large number of clients subscribe to a common content provider that frequently pushes multimedia content to the subscribers, e.g., text, photos, or videos. This can potentially generate thousands of duplicated downloads of the same content thus consuming a great amount of bandwidth in cellular systems~\cite{wang2015social} and~\cite{bastug2014living}. As a result of both local area and social network services, a large part of the cellular traffic consists of a few popular files that must be delivered to co-located social groups of user equipments (UEs).

Offloading the traffic of local area services by leveraging direct device-to-device (D2D) communication links between UEs, is a promising solution to reduce the congestion on existing cellular networks~\cite{zhao2015social,Mine1,Unmanned_Mozaffari}. In a D2D offloading procedure, the users can receive data from other UEs over D2D links instead of using the cellular links~\cite{Mine1,andreev2014cellular}. Since a large amount of traffic is generated by a few popular contents, caching popular files on the users' devices is one promising solution to offload the cellular traffic and reduce the load on the base stations and backhaul~\cite{chen2016cooperative}. Thus, caching popular files at the UE level and disseminating it via the use of D2D communication links is now seen as a key approach for boosting the performance of tomorrow's 5G networks \cite{hamidouche2016mean}. To properly decide on how and where to cache content, one must take into account, not only wireless physical parameters, such as channel gain or interference, but also new user-specific information, such as social metrics or geolocation, as discussed in~\cite{wang2015social,zhao2015social} and~\cite{hamidouche2016mean}. Indeed, both the channel quality over the D2D links and the social tie between UEs become important to decide on how to place content and share the cached data. On the one hand, the social tie determines the common interests of the users, thus determining the way in which they share content. On the other hand, the data rate over the D2D links will determine the effectiveness of the data sharing~\cite{wang2015social}. For instance, in local area and social network services, the content provider first sends the content to a target set of UEs (known as \emph{seeds}) via cellular network links. These seed UEs then cache the content and use D2D communication links to share it with other UEs in proximity and that belong to various social communities.

Many recent works have focused on developing new techniques to offload social network traffic by exploiting D2D communications among UEs such as~\cite{chen2016cooperative,semiari2015context,abdel2012energy}, and~\cite{li2014social}. Some of these works such as in~\cite{chen2016cooperative} have mainly focused on managing the interference among D2D links to increase the cooperative opportunity in sharing cached content at the UEs side. In contrast, in~\cite{abdel2012energy}, the authors focused on optimizing the overall system performance (e.g., minimize bandwidth) or the average delay of subscribers. However, the works in~\cite{chen2016cooperative} and~\cite{abdel2012energy} consider that the UEs always offload their cached content over the D2D links, and they do not consider the social tie between UEs. The works in~\cite{kempe2003maximizing} and~\cite{kempe2005influential} have focused on the influence maximization problem in social networks. The influence maximization problem is approximated in general models that are referred to as the decreasing cascade and linear threshold models. Then, a greedy algorithm is proposed to choose a set of individuals such that the initial activating of this set is as large as possible in influence expectation. In~\cite{alim2017leveraging}, a novel framework is proposed to enable devices to form multi-hop D2D connections in an effort to maintain sustainable communication in the presence of device mobility. The framework proposed in~\cite{alim2017leveraging} can be used to derive an optimal solution for time-sensitive content transmission while also minimizing the cost that the base station pays in order to incentivize users to participate in D2D. The work in~\cite{nguyen2014dynamic} focuses on how to efficiently identify communities in dynamic social networks. In this work, the authors present quick community adaptation, an adaptive modularity-based framework for not only discovering but also tracing the evolution of network communities in dynamic online social networks. {In~\cite{R2_11}, the problem of optimally determining source$\text{-}$destination connectivity of random networks with a finite number of nodes is studied. The authors in~\cite{R2_11} determine a policy for establishing whether a designated source and destination are connected with minimum expected cost. The proposed policy in~\cite{R2_11} simply condenses each known connected component to a single super node at each step, and in that condensation multi-graph it simply tests an edge that is both on the shortest path containing the super nodes, as well as on a minimum source$\text{-}$destination cut.}

In~\cite{semiari2015context,li2014social,Zhang2017}, and~\cite{Bai2016}, a social-aware D2D communication architecture is proposed to leverage social network features to optimize the use of D2D communications. For example, social ties can measure the strengths of users in D2D systems, and reflect to some degree the communication demands between UEs~\cite{semiari2015context}. Moreover, a high degree of centrality in a social network can imply that a given user may play a key role in data transmission. Indeed, users usually share popular cached content to each other in the D2D network only if they have a strong enough social tie~\cite{wang2015social} and they may participate in different social communities. As a result, UEs with high centrality should be allocated more wireless resources so as to leverage their connections in D2D transmission~\cite{li2014social}. In~\cite{Zhang2017}, a social-aware framework for optimizing D2D communications is presented by exploiting users' relationships in the social network, and connections of UEs in the physical wireless network. Then, to enhance cooperation of users in content delivery in D2D network, the authors have proposed to use different social networking features. In~\cite{Bai2016}, a novel hypergraph framework is proposed to for studying social-aware caching in D2D networks. In particular, the authors use different hypergraph concepts, such as hypergraph coloring and multidimensional matching, to optimize spectrum allocation and cachce placement in a D2D network. In~\cite{Zhuang2016}, the authors study the information diffusion in a clustered multilayer network model, where all constituent layers are random networks with high clustering. One of the key results of~\cite{Zhuang2016} is that information with low transmissibility spreads more effectively within a small but densely connected social network. In~\cite{fu2017anonymization}, the authors present a comprehensive study of the community-structured social network de-anonymization problem. The main focus of this work is on privacy and anonymization challenges. {In~\cite{R2_51}, the authors prove that most properties of nodes, links, and paths are correlated among the social and D2D graphs. Then, they use the structure of the social graph to build forwarding paths in the D2D graph, allowing two nodes to communicate over time using opportunistic contacts and intermediate nodes. In~\cite{R2_52}, the authors present a set of new temporal distance based metrics. Then, they show how these metrics can be applied effectively to characterise the temporal dynamics and data diffusion efficiency of social networks.} {In~\cite{R2_12}, the throughput capacity of wireless networks with social characteristics is studied. In particular, the proposed model in~\cite{R2_12} captures the impacts of the way people choose friends as well as the number of friends on the capacity of real large-scale networks specifically for the multicast traffic.} {In~\cite{R_minor2}, the authors propose a novel approach to detect properties of social grouping and human mobility. Then, popular social network users are used for one-hop opportunistic data forwarding. The work in~\cite{R_minor1} introduces new policies for dividing large communities into sub-communities following location or social interests. Then, users known as multi-homed users are exploited deliver data across the sub-communities.}

In practical social networks, users may belong to different social communities where each community's members have the same interests in receiving cached content. Thus, in a D2D network, the social tie among users of one community, centrality in each social community, and the effects of different communities on on another must be considered to improve data offload via caching and D2D links. None of these critical parameters are accounted for in the existing works such as~\cite{wang2015social,chen2016cooperative,semiari2015context,abdel2012energy,hamidouche2016mean,li2014social,kempe2003maximizing,kempe2005influential}, and~\cite{Zhang2017}. Moreover, even though the majority of existing literature such as ~\cite{wang2015social,chen2016cooperative,semiari2015context,abdel2012energy,hamidouche2016mean,kempe2003maximizing,kempe2005influential}, and~\cite{Zhang2017} focuses on single community, some works such as ~\cite{alim2017leveraging,nguyen2014dynamic,li2014social,Zhuang2016}~and~\cite{fu2017anonymization}, do consider multiple communities. However, these multi-community works do not capture the dependence and effect of the centrality of one social community on the other, which is particularly important for cache placement in real-world D2D networks.

The main contribution of this paper is a new framework to optimally select a suitable set of seed UEs that can be used for the cache placement over D2D network. The proposed framework leverages multi-community social network features to optimize multi-hop D2D offloading procedure. Indeed, the proposed framework allows to maximize the expected number of UEs that receive cached content from the cache placement set through multi-hop D2D offloading procedure. To quantify the collaborative effect of seed UEs (in the cache placement set) on local data offload, a cooperative game in characteristic function form is proposed. For this game, we prove that the Shapley value of each UE in the cache placement set captures the exclusive effect of this UE on the effectiveness of offloading popular content over D2D links. To capture the effects of multiple communities of users on each other in a D2D network, we model the social graphs among the users and the D2D graph among the UEs using a hypergraph. Then, we propose two line graphs: directed influence-weighted graph and directed connectivity-weighted graph for analyzing the hypergaph model. Using the combination of the Shapley value and the hypergraph model, we define a new offload power metric for the UEs. This metric quantifies the power of each UE in offloading a cached content over the multi-community multi-hop D2D network. Simulation results show that, on the average, the proposed framework achieves $12\%$, $19\%$ and $21\%$ improvements in terms of the number of UEs that received offloaded content compared to the schemes based on betweenness, degree, and closeness centrality, respectively.

The rest of this paper is organized as follows. Section~\ref{Dec:Sys-Model} presents the system model. In Section~\ref{Network Centrality}, the network centrality problem for optimal seed selection in social multi-hop D2D links is formulated. Then, in Section~\ref{Cooperative Game}, a cooperative game approach for solving the network centrality problem is proposed and the properties of its Shapley value are studied. Then, our framework for cache placement based on the hypergraph model and the Shapley value approach is presented. In Section~\ref{Dec:Complexity}, the complexity of our proposed approach is analysed. In Section~\ref{Dec:Simulation}, we provide the simulation results while conclusions are drawn in Section~\ref{Dec:Conclusion}.
\vspace{-0.4cm}
%*****************************************************************
\section{System Model}
\label{Dec:Sys-Model}
Consider a multi-hop D2D-enhanced cellular network in which a set $\mathcal{N}$ of $N$ wireless user equipments can communicate directly via D2D communication links. Each user can access popular content from a base station (BS) over a cellular link or from a seed over a multi-hop D2D link. In our model, the network operator always ensures that the content cached at the seed is fresh and corresponds to the most popular content. Thus, most of the time, a user can obtain a fresh popular content from a seed. In case the seed does not have the requested content, then, the user can download it directly from the BS. In this network, a given UE can deliver a file of size $B$  bits to a neighbor, i.e., in one hop, within one time slot $t$. The duration of each time slot is $T$ seconds. Our model focuses on delay-tolerant services that are not affected by the potential delay incurred by multi-hop transmissions.

We assume an overlay D2D communication model ~\cite{asadi2014survey}, in which a portion of the cellular resources is dedicated to D2D communications. Hence, no mutual interference occurs between D2D and cellular links. Consequently, the interference over any D2D link between two UEs $m$ and $n$ depends on other D2D pairs that communicate over the same resource block (RB) assigned to the D2D link between UEs $m$ and $n$. We consider an orthogonal frequency division multiple access scheme for the D2D transmissions. In this scheme, each D2D link will be assigned one RB. We assume that the transmission power of each UE $m$ is $p_m$ and the bandwidth of each resource block on the D2D link is equal to $B_w$. Consequently, the data rate between UE $m$ and UE $n$ is given by:
\begin{equation}
R_{mn}= B_w\log\left(1+\frac{\beta p_m h_{mn}}{B_wN_0+\sum_{k} h_{kn} p_k}\right),
\end{equation}
where $h_{mn}$ is the channel gain between UE $m$ and UE $n$ on each RB at time slot $t$, $N_0$ is the noise power spectral density, and $\beta=\frac{-1.5}{\ln{5P_e}}$ is the SNR gap for M-QAM modulation with $P_e$ being the maximum acceptable error probability. $\sum_{k} h_{kn} p_k$ is the interference from any other UE $k\neq m$ on the D2D link between UE $m$ and $n$ when the same RB is allocated to the UEs $k$ and $m$. We also assume a block fading channel for the D2D links whose fading process is assumed to be constant during one time slot (i.e., $T$ seconds). Consequently, we can consider a constant bit rate over each D2D link during one time slot.

Considering all D2D links among UEs in the communication network, we introduce a D2D graph $G^d(\mathcal{N}, \mathcal{E}_d)$ whose set of vertices is the set $\mathcal{N}$ of UEs and the set of edges (links) is  $\mathcal{E}_d=\{(m,n)| m,n\in\mathcal{N}\text{ and }\frac{B}{R_{mn}}\leq T\}$. Thus, a link exists from a given vertex $m$ to another vertex $n$ in the D2D graph if and only if UE $m$ can transmit a single $B$-bit packet to UE $n$ during one time slot $t$ of duration $T$ over a direct D2D link. Since UEs are carried by human users, we assume that all of the $N$ users form a multi-community social network. In this social network, there are $L$ social communities connecting the users who are carrying the UEs that form the D2D network. Let $\mathcal{L}_l$ be the set of UEs belonging to social community $l$, thus $\cup_{l=1}^L \mathcal{L}_l=\mathcal{N}$.

Each social community $l$ is modeled by a weighted social graph $G^s_l(\mathcal{L}_l,\mathcal{E}_l^s,w_l^s)$, whose vertices are the UEs belonging to social community $l$ and whose edges are given by the set $\mathcal{E}_l^s=\{(m,n)| 0 < w_{mn} , \forall m,n \in \mathcal{L}_l\}$, where $w_{mn}$ is the social tie between UEs $m$ and $n$ which is obtained using the function $w_l^s:\mathcal{E}_l^s\rightarrow (0,1]$. This function captures the strength of the social tie. A higher value of $w_{mn}$, $w_l^s:(m,n)\rightarrow w_{mn}$, represents a stronger social tie between UEs $m$ and $n$.

Social ties are used to capture social relationships between users such as: friendship, kinship, colleague relationships, and altruistic behavior that are observed in human activities~\cite{li2014social}. Due to the social ties among members of each community, the users in each community exhibit homophily as they share common contents~\cite{li2014social}. For example, students usually share content related to their major or fans of specific sport will tend to share news about it. Thus, we assume that all members of each social community are interested in a common popular content. Consequently, we consider $L$ popular files, each of which corresponds to one community. { Since number of UEs that can cache the popular file for each community can be seen as a budget, we consider the worst-case scenario where just one UE from each community is selected as a \emph{seed}}. We defined a seed set $\mathcal{S}_0$ consisting of $L$ seed UEs. The BS sends to each seed the popular file that corresponds to each community. Then, the seed caches the popular content. {The popular file per community needs to be received by all of the members of that community. For a given community and its associated popular file, other users belonging to other communities can help to forward the given popular content to the UEs of the given community by multi-hop D2D transmission.} Since each UE $m$ can send its cached content to another UE $n$ over multi-hop D2D links, we say that UE $n$ can be influenced by another UE $m$. Here, influence means receiving content over D2D links from other UEs that locally cached the data. Accordingly, we define the following \textit{one-hop influence} concept.
\begin{definition}\label{def-onehop}
\textnormal{The \textit{one-hop influence} of a given UE $m$ on its one-hop neighbor UE $n$ in the D2D graph is the preference of UE $m$ to transmit its cached content to UE $n$ over a direct one-hop D2D link .}
\end{definition}

Note that the defined \textit{one-hop influence} depends on the social tie between two UEs of each community and also the multi-hop D2D path between them. In other words, when the social tie between members increases then the probability of transmitting a cached content to the neighbors will also increase. Moreover,  when a node locally shares its cached content over D2D links, the content may go from one community's members to another over the D2D graph.

For example, assume that a given UE $m\in \mathcal{L}_u$ must send its cached content to another member of its community. It will then send this content to UE $n$ that belongs to another community $\mathcal{L}_v$, if the shortest path between UE $m$ and another member in community $\mathcal{L}_u$ includes the UE $n$ over the D2D graph.

Thus, if UE $m\in \mathcal{L}_u$ and UE $n$ are neighbors in the D2D graph $G^d$, the \emph{one-hop influence} of UE $m$ on the UE $n$ will be:
\begin{equation}
I_{mn}=\sum_{m'\in \mathcal{L}_u\backslash\{m\}, n\in\mathcal{P}^d_{mm'}}\frac{w_{mm'}}{|\mathcal{P}^d_{mm'}|},
\label{Dmodel}
\end{equation}
where $\mathcal{P}^d_{mm'}$ is the set of UEs which form a shortest path from UE $m$ to UE $m'$ within the D2D graph $G^d$.

Given the one-hop influence model in (\ref{Dmodel}), we consider the influence graph as a weighted directed graph $G^i(\mathcal{N},\mathcal{E}_d,w_i)$. In $G^i$, vertices are the set $\mathcal{N}$ of UEs and the weight of each edge $(m,n) \in \mathcal{E}_d$ captures the one-hop influence of UE $m$ on UE $n$, i.e., $I_{mn}$.

We illustrate these parameters using a simple example, shown in Fig.~\ref{DModel}. In this example, 10 UEs are partitioned into two social communities $G_1^s$ and $G_2^s$ and form one D2D graph $G^d$. For instance, since the channel gain between UE 2 and UE 3 is high enough, they are connected via a D2D link $(2,3)\in \mathcal{E}_d$. Due to the social tie between UE 3 and UE 5 in social community $2$, $(3,5)\in \mathcal{E}_2^s$ and the weight $w_{35}$ captures the social tie between UE $3$ and $5$. For the influence graph $G^i$ in this example, we must compute the one-hop influence between all D2D neighbors. For example, consider UEs 2 and 3 in the D2D graph. These UEs belong to different communities but they are neighbors. Using Definition~\ref{def-onehop}, we need to calculate $I_{23}$ and $I_{32}$ in order to obtain the influence graph $G^i$ as shown in the Fig.~\ref{DModel}. To compute $I_{23}$, we need to consider that UE 2 would like to share a content with UE 6 with which it has social tie $w_{2,6}$ in $G^s_1$. Given that UE 3 is in its shortest path toward UE 6 in the D2D graph, we can calculate $I_{23}$ as $\frac{w_{2,6}}{3}$. Similarly, to calculate $I_{32}$, we need to consider that UE 3 would like to share a content with UEs 1 and 5 with which it has social ties $w_{3,1}$ and $w_{3,5}$ in $G^s_2$. Given that UE 2 is not in the shortest path toward UE 5 and is only in the shortest path toward UE 1 in the D2D graph, we can calculate $I_{32}$ as $\frac{w_{3,1}}{2}$. Following the same procedure for all neighbors in the D2D graph, we can obtain the influence graph $G^i$ of the D2D graph. As shown in Fig.~\ref{DModel}, UEs 6 and 5 are selected as seeds and the BS sends popular content 1 to the UE 6 in social community 1 and popular content 2 to the UE 5 in social community 2. The list of notations used throughout this paper is presented in Table~\ref{tab:Symbols}.
\begin{figure}[t]
\centering
\includegraphics[scale=0.35]{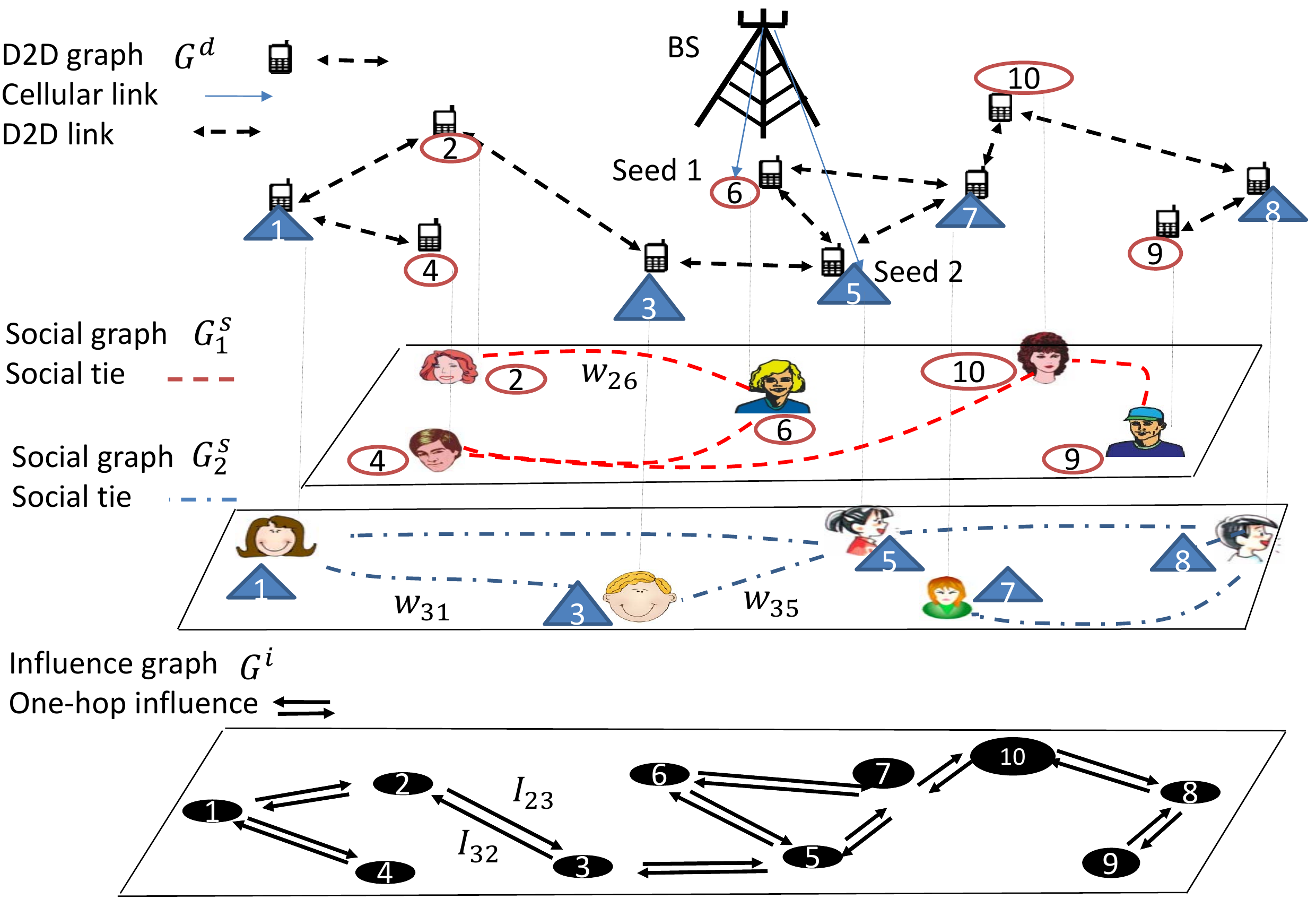}
\caption{An illustrative example of D2D graph and social network for 10 UEs.}
\label{DModel}
\vspace{-.3cm}
\end{figure}

\begin{table}[t]
\caption{List of notations used throughout the paper.}
\vspace{-.5cm}
\begin{center}
%\begin{adjustbox}{width=.95\columnwidth}
\begin{tabular}{l | l } \toprule
{ \textbf{Symbol}}  & { \textbf{Definition}}   \\  \rowcolor[gray]{.9} \hline
$\mathcal{N}$   &  Set of all UEs\\
$T$ &  Duration of one time slot $t$\\ \rowcolor[gray]{.9}
$R_{mn}$ & Bit rate between UE $m$ and UE $n$ on
each RB \\
$\mathcal{E}_d$ & Set of all D2D links\\ \rowcolor[gray]{.9}
$G^d$ & D2D graph\\
$\mathcal{L}_u$   & Set of UEs in social community $u$\\ \rowcolor[gray]{.9}
$\mathcal{E}_l^s$ & Set of all social relationships in community $l$\\
$w_{l}^s$   & The function: $\mathcal{E}_l^s\rightarrow[0,1]$. \\ \rowcolor[gray]{.9}
$w_{mn}$   & Social tie between UE $m$ and $n$ \\
$G_l^s$   & Weighted social graph of community $l$\\ \rowcolor[gray]{.9}
$\mathcal{P}_{mm'}^d$ &  Shortest path set of UEs from UE $m$ to $m'$ \\\rowcolor[gray]{.9}
& on D2D graph\\
$\mathcal{C}_m$   & Set of one-hop neighbor of UE $m$ in the D2D graph\\ \rowcolor[gray]{.9}
$\mathcal{C}_{m,d}$   & Set of $d$-distance neighbor of UE $m$ in the D2D graph\\
$d_{mn}$ & Length of the shortest path between the UE $m$ and $n$ \\
& in the influence graph \\ \rowcolor[gray]{.9}
$G^i$  & Weighted influence graph\\
$I_{mn}$   &One-hop influence of UE $m$ on the UE $n$\\  \rowcolor[gray]{.9}
$\mathcal{S}_0$   & Seed set\\
$\mathcal{S}_t$ & Set of UEs that received the social content by time slot $t$ \\ \rowcolor[gray]{.9}
$\mathfrak{S}$  & Diffusion process\\
$I_{d_n}(\mathcal{S}_t)$   &$d_n$-\textit{influence} of $\mathcal{S}_t$ on UE $n\in \mathcal{N}\backslash \mathcal{S}_t$\\  \rowcolor[gray]{.9}
$H$   &Hypergraph\\
$I_{d_n}(\mathcal{S}_t)$   &Exclusive influence of UE $k$ on $\mathcal{C}_{n,d}$ of UE $n$\\  \rowcolor[gray]{.9}
$v(\mathcal{S}_0,G^d)$ &Value of coalition $\mathcal{S}_0$ in graph $G^d$\\
$\phi_k(\mathcal{S}_0,G^d)$   & Shapley value of player $k$ in coalition $\mathcal{S}_0$\\  \rowcolor[gray]{.9}
$\phi_k(G^d)$  & Shapley value of UE $k$ in graph $G_d$\\
$O_k$   & Offloading power of UE $k$\\  \rowcolor[gray]{.9}
$D_i(H)$   &Directed influence-weighted line graph of hypergraph $H$\\
$D_c(H)$   &Directed connectivity-weighted line graph of  $H$\\  \rowcolor[gray]{.9}
$\mathfrak{S}$  & Set of all social communities\\
$\mathcal{E}_H$   &Edge set of line graph $D_i(H)$ or $D_c(H)$\\  \rowcolor[gray]{.9}
$w_i$   &Directed connectivity-weighted line graph of $H$\\
$w_c$  & Set of all social communities\\
\hline
\end{tabular}
%\end{adjustbox}
\end{center}
\label{tab:Symbols}
\vspace{-9mm}
\end{table}

In general, the D2D links between UEs as well as their social tie will affect the total number of UEs that cache and offload the popular content over D2D links. Thus, we need to exploit the social tie among users to select an appropriate seed set in the D2D graph. The optimal seed set can be selected to maximize the total expected number of UEs that will ultimately receive the popular content over a multi-hop D2D network. Finding the optimal seeds in a directed weighted graph is known as a top-$k$ node problem or influence maximization problem~\cite{gomez2003centrality}.

Our goal is to select a seed set of $k$ users to maximize the expected number of UEs that will receive the cached data in the D2D graph. The network can then send popular content directly to the seed set over cellular links from the BS. Then, UEs in the seed set cache the popular content and other UEs can download this content from the the seed set over D2D links~\cite{zhang2014recent}. This optimal set is called the most influential set of nodes in a network. Such a top-$k$ node or influence maximization problem is known to be NP hard~\cite{zhang2014recent,dhamal2014cooperative}. Next, we exploit the social ties among members of one community and the effect of members of communities on each others over the multi-community social aware multi-hop D2D graph to derive the optimal seed set of UEs.
\vspace{-0.4cm}
%*****************************************************************
\section{Network Centrality for Seed Selection}
\label{Network Centrality}
Our main goal is to find the center of the D2D graph $G^d$. This will allow finding the optimal cache placement, in order to maximize the expected number of UEs that receive cached content from the cache placement set using multi-hop D2D sharing. Then, the UEs can receive data from the seed set members and distribute it over the D2D graph according to their social tie. Let $\mathcal{S}_t$ be the set of UEs that have received popular content until time slot $t$. Then, $\mathfrak{S}=\{\mathcal{S}_0,\mathcal{S}_1,...,\mathcal{S}_t\}$ is defined as a \emph{diffusion process} in which $\mathcal{S}_t$ is the set of UEs that have received content and cached it by the end of time slot $t$. The influence maximization problem for offloading social data in multi-hop D2D networks can be defined as follows:
\begin{definition}\label{def:influence problem}
\textnormal{The \textit{influence maximization problem} in $L$-community multi-hop D2D networks aims to select a seed set $\mathcal{S}_0$ consisting of $L$ UEs, to maximize the number of UEs that received cached content over D2D links.
%
%where .  including $L$ UEs to send them popular content over the cellular links such that the expected number of UEs that received cached content on D2D links from seed set can be maximized.
}
\end{definition}

Let $\mathcal{C}_n$ be the set of one-hop neighbors of UE $n$ within graph $G^d$. We define the distance $d_{mn}$ between two UEs $m$ and $n$ as the summation of the links' weights in the shortest path between the UEs in the weighted influence graph $G^i$.

Then, we define $\mathcal{C}_{n,d}=\{m|m\in G^d, d_{mn}\leq d\}$ to be the set of $d$-distance neighbors of UE $n$. This set includes the UEs whose distance from UE $n$ is less than $d$. A lower distance between UEs leads to faster offloading of the cached content through D2D sharing. Next, we define the following concept.
\begin{definition}
\textnormal{The $d$-\textit{influence} of $\mathcal{S}_t$ on UE $n\in \mathcal{N}\backslash \mathcal{S}_t$ is defined as the expected number of UEs in the set $\mathcal{C}_{n,d}$ that can receive the cached data over the D2D graph from the UEs in $\mathcal{S}_t$. This is given by:}
\end{definition}
\begin{equation}
I_{d}(n,\mathcal{S}_t)=\sum_{j \in \mathcal{C}_{n,d}}\biggl(1-\prod_{\substack{m\in \mathcal{C}_j \cap \mathcal{S}_t}} (1-I_{mn})\biggr),
\label{d_influence}
\end{equation}
where $I_{mn}$ is given by (\ref{Dmodel}).

Whenever a UE $k\in \mathcal{S}_t$ transmits its cached data to one of the UEs in the $d$-distance neighbor set of UE $n$, we say that UE $k$ affects UE $n$. If all UEs in $\mathcal{S}_t$, except UE $k$ fail to affect UE $n$, then we can say that UE $k$ exclusively affects UE $n$. Now, we can calculate the exclusive influence of each UE $k$.
\begin{proposition}
\textnormal{The \textit{exclusive influence} of each UE $k\in \mathcal{S}_t$ on the $d$-distance neighbor set of a UE $n$ is given by:}
\begin{equation}
I_{d}(n,k)=I_{d}(n,\mathcal{S}_t)-I_{d}(n,\mathcal{S}_t \backslash \{k\})
\label{per_influence}
\end{equation}
\end{proposition}
\begin{proof}
Given (\ref{d_influence}), we can write:
\begin{multline*}
I_{d}(n,\mathcal{S}_t)-I_{d}(n,\mathcal{S}_t \backslash \{k\})=\\
\sum_{j \in \mathcal{C}_{n,d}}\biggl(1-\prod_{\substack{m\in \mathcal{C}_j \cap \mathcal{S}_t}} (1-I_{mn})\biggr)-\\
\sum_{j \in \mathcal{C}_{n,d}} \biggl( 1-\prod_{\substack{m\in \mathcal{C}_j \cap \{ \mathcal{S}_t\backslash \{k\}\} }} (1-I_{mn})\biggr)=\\
\sum_{j \in \mathcal{C}_{n,d}} \biggl( \prod_{\substack{m\in \mathcal{C}_j \cap \{\mathcal{S}_t \backslash \{k\}\} }} (1-I_{mn})
-\prod_{\substack{m\in \mathcal{C}_j \cap \mathcal{S}_t}} (1-I_{mn})\biggr)=\\
\sum_{j \in \mathcal{C}_{n,d}}\biggl((1-\prod_{\substack{m\in \mathcal{C}_j\cap\{k\}}} (1-I_{mn}))
\prod_{\substack{m\in \mathcal{C}_j \cap \{\mathcal{S}_t \backslash \{k\}\} }} (1-I_{mn})\biggr),
\end{multline*}
where $(1-\prod_{\substack{m\in \mathcal{C}_j\cap\{k\}}} (1-I_{mn}))$ represents the probability that UE $k$ shares the cached data with at least one of the one-hop neighbors of UE $j$, where $ j \in \mathcal{C}_{n,d}$. Moreover, $\prod_{\substack{m\in \mathcal{C}_j \cap \{\mathcal{S}_t \backslash \{k\}\}}} (1-I_{mn})$ is the probability that none of the members of $S_t \backslash \{k\}$  can share the content with one of the one-hop neighbors of UE $j$, where $ j \in \mathcal{C}_{n,d}$. Thus, the multiplication of these two terms represents the probability that UE $k$ exclusively offloads its cached content to one of the members of   $\mathcal{C}_{n,d}$. Considering the summation over all the members of $\mathcal{C}_{n,d}$, the above expression represents the  expected number of members in $\mathcal{C}_{n,d}$ to which UE $k$ exclusively offloads its cached data.
\end{proof}
\vspace{-0.2cm}
According to Definition~\ref{def:influence problem}, the influence optimization problem that maximizes the expected number of UEs that received data over D2D links and cached it, is given by:
\begin{align}\label{opt_problem}
&\max_{\mathcal{S}_0} \lim_{t\rightarrow \infty} \mathds{E} \big( |\mathcal{S}_t| \big).
\end{align}
\vspace{-0.1cm}
The effect of the initial seed set $\mathcal{S}_0$ on another future set $\mathcal{S}_t$ in the diffusion process $\mathfrak{S}$ depends on the $d_n$-influence of $\mathcal{S}_t$ in each time slot. Generally, an influence maximization problem such as (\ref{opt_problem}) is known to be NP-hard~\cite{zhang2014recent} and~\cite{kempe2003maximizing}.

There are some conventional sub-optimal solutions for solving (\ref{opt_problem}) such as degree centrality, betweenness centrality, and closeness centrality~\cite{koschutzki2005centrality}. However, they are not sufficient to properly understand the relative importance of seeds as stand-alone entities in the D2D graph. The reason is that these existing methods select the seed set without considering the cumulative effect that the selected seeds have on each other during the diffusion of social content over the D2D graph~\cite{aadithya2010efficient} and~\cite{liu2010modeling}. However, in a multi-community social network that uses a multi-hop D2D network, the effect of each seed in offloading social data depends on the contribution of other seeds. This is due to the physical connection between the UEs with in the D2D network and the common social interests of the users in a given community. Consequently, we need to take into account the contributions of all possible combinations of UEs from different communities in offloading content through D2D links.

Game-theoretic network centrality has recently attracted attention as a promising solution to address the above limitation such as in~\cite{gomez2003centrality}. In particular, it has been shown that the cooperative game concept of a Shapley value (SV) is an effective measure of the importance of players within a group~\cite{gomez2003centrality} and~\cite{narayanam2011shapley}. In fact, the SV of each node in a graph can be considered as its influence when combined with other nodes if the value function of the cooperative game is appropriately defined as the influence of the players on other nodes over a graph. As a result, the Shapley value of each UE in the a given game can be interpreted as a centrality measure. The use of SV-based network centrality confers a high degree of flexibility (which is completely lacking in traditional centrality metrics) to capture the social tie and wireless channel gain between UEs. Moreover, this new paradigm has already been proved to be more useful than traditional centrality measures for certain real-life network applications such as in~\cite{narayanam2011shapley}. Thus, next, we will define a game and compute the Shapley value to sub-optimally solve the centrality problem (\ref{opt_problem}).
\vspace{-0.4cm}
%*****************************************************************
\section{{Seed Selection based on a Cooperative Game}}
\label{Cooperative Game}

Given a graph $G^d$, we use $g$ to define a coalitional game $g(\mathcal{N},v)$ whose set of players is the set of UEs $\mathcal{N}$ (in the D2D graph). Here $v$ is the game's \textit{characteristic function}~\cite{young2014handbook}. A coalition of UEs $\mathcal{S}$ is simply any subset of $\mathcal{N}$. The value of a given coalition $\mathcal{S}$, which is a function over the real line, depends on the D2D graph, $v(\mathcal{S},G^d)\rightarrow \mathbb{R}$. Thus, considering a coalition $\mathcal{S}$ of UEs, the definition of the characteristic function must quantify the effect of this coalition on offloading content over the D2D graph $G^d$~\cite{aadithya2010efficient}. Given the diffusion model in~(\ref{Dmodel}), the influence of each UE in coalition $\mathcal{S}$ on its one-hop neighbors, the number of one-hop neighbors in the current time slot, and the possible future multi-hop neighbors during next time slots affect the diffusion process $\mathfrak{S}$. Thus, we define a value function for the game that reflects these parameters in its characteristic function formulation.
\vspace{-0.35cm}
\subsection{Value Function}
Given a D2D graph $G^d$, we define the value of the game as the summation of the $d_n$-influence of its members on the UE in the D2D graph, as follows:
\begin{equation}
v(\mathcal{S}_0,G^d)=
  \begin{cases}
      \sum\limits_{n\in N\backslash \mathcal{S}_0}\alpha_n I_{d_n}(\mathcal{S}_0)& \quad \text{if $\mathcal{S}_0\neq\O$}\\
    0 & \quad \text{else}.\\
  \end{cases}
\label{value function}
\end{equation}
%\begin{equation*}
%v(\mathcal{S}_0,G^d)=
%\end{equation*}
%\begin{multline}
%\begin{cases}
%      \sum\limits_{n\in N\backslash \mathcal{S}_0}\alpha_n I_{d}(n,\mathcal{S}_0)& \quad \text{if $\mathcal{S}_0\neq\O$}\\
%    & \quad \text{otherwise}.
%  \end{cases}
%\label{value function}
%\end{multline}

A coalition that achieves a higher value function in (\ref{value function}) will have a higher probability of sending the cached data to the other UEs over the D2D graph. Here, $\alpha_n$ is a price parameter per unit influence. Thus, the value function in (\ref{value function}) will be a monetary value and of \textit{transferable utility}~\cite{young2014handbook}. From (\ref{d_influence}), we can see that the value of each coalition is related to the social tie between its members and their one-hop neighbors over the D2D graph.
\vspace{-0.35cm}
\subsection{Seed Selection using the Shapley Value}
The Shapley value , $\phi_k(\mathcal{S}_0,G^d)$, of a player $k$ in a coalition $\mathcal{S}_0$ is given by $\phi_k(\mathcal{S}_0,G^d)=\sum\limits_{\mathcal{R}\subseteq\mathcal{S}_0\backslash\{k\}}\frac{(|\mathcal{S}_0|-|\mathcal{R}|-1)!|\mathcal{R}|!}{|\mathcal{S}_0|!}\big(v(\mathcal{R}\cup\{k\},G^d)-v(\mathcal{R},G^d)\big)
$~\cite{young2014handbook}. Consequently, if a UE $k$ achieves a high SV, this UE will contribute more to the value function of any randomly chosen coalition of UEs $\mathcal{R}$ compared to UEs in $\mathcal{S}_0$. Following (\ref{value function}), the defined value function captures the influence of coalition $\mathcal{S}_0$ in the distribution of the cached data over $G^d$. Thus, a higher SV for a UE $k\in \mathcal{S}_0$ implies a higher degree of collaboration between this UE and other UEs in $\mathcal{S}_0$ for the purpose of transmitting cached content over D2D links. Next, we show how the Shapley value of each UE $k$ is related to its exclusive influence on the UEs which are not in its coalition.
\begin{theorem}
\textnormal{The Shapley value of each UE $k$ in coalition $\mathcal{S}_0$, is equal to the exclusive influence of UE $k$ on the UEs which are not in $\mathcal{S}_0$, which is given by:
\begin{equation}
\phi_k(G^d)=\sum\limits_{n:{ \{ \mathcal{C}_k\cap \mathcal{C}_{n,d} \} \neq \emptyset}}\frac{\alpha_n}{1+|\mathcal{C}_{n,d}|}.
\label{shaply_value_CF}
\end{equation}}
\end{theorem}
\begin{proof}
Following (\ref{per_influence}) in Proposition 1, we can write $\phi_k(\mathcal{S}_0,G^d)=\sum\limits_{\mathcal{R}\subseteq \mathcal{S}_0\backslash\{k\}}\frac{(|\mathcal{S}_0|-|\mathcal{R}|-1)!|\mathcal{R}|!}{|\mathcal{S}_0|!}\sum\limits_{n\in \mathcal{N}\backslash \mathcal{R}}I_{d}(n,k)$. Given a coalition $\mathcal{R}$ and a UE $k \notin \mathcal{R}$, the necessary condition under which UE $k$ exclusively affects another UE $n$ to have an empty intersection set between the one-hop neighbor sets of all UEs in coalition $\mathcal{R}$ and the $d$-distance neighbor set of UE $n$. This implies that $\{\cup_{j\in \mathcal{R}} \mathcal{C}_j\}\cap \mathcal{C}_{n,d} = \emptyset$. Given that the permutations are chosen uniformly for computing the SV, it has been shown in~\cite{aadithya2010efficient} that this necessary condition, $\{\cup_{j\in \mathcal{R}} \mathcal{C}_j\}\cap \mathcal{C}_{n,d} = \emptyset$, is satisfied with probability $\frac{1}{1+|\mathcal{C}_{n,d}|}$. Thus, $\text{Pr}(\{\cup_{j\in \mathcal{R}} \mathcal{C}_j\}\cap \mathcal{C}_{n,d} = \emptyset)=\frac{1}{1+|\mathcal{C}_{n,d}|}$. Moreover,  if UE $k$ wants to send its cached content to the set of $d$-distance neighbors of UE $n$, at least one of the members of the $d$-distance neighbor set of UE $n$ must be in the set of one-hop neighbors of UE $k$. This means that $\mathcal{C}_k\cap \mathcal{C}_{n,d}\neq \emptyset$. Thus, the Shapley value of UE $k$ in the social weighted D2D graph $G^d$ is given by (\ref{shaply_value_CF}).
\end{proof}
The relationship in (\ref{shaply_value_CF}) shows that the SV of a given UE will be affected by two key factors: a) the number of its one-hop neighbors and b) the UEs in the $d$-distance of its one-hop neighbors. In other words, the Shapley value of UE $k$ in a coalition will be higher if UE $k$ has many one-hop neighbors and the distance of these one-hop neighbors from other members in the coalition is less than $d$. Consequently, if we select the seed set according to the Shapley value of UEs for the games defined in (\ref{value function}) over the graph, each seed can send its cached content to those UEs that are not in the $d$-distance of the one-hop neighbor set of other seeds.
%The centrality literature such as~\cite{aadithya2010efficient} and~\cite{bonacich1987power} also recognize this fact.

If we just model the interactions of UEs on the simple D2D graph with the proposed game (\ref{value function}), we will not capture the effect of two key parameters on the D2D sharing of the cached content: a) the social tie between members of each community and b) the effect of the communities on one another. One natural way to capture these interactions among different members of the social communities that interact over the D2D graph is using a hypergraph model~\cite{roy2015measuring} and~\cite{liu2010modeling}.
\vspace{-0.35cm}
\subsection{Social Communities as Hypergraphs}
If the problem of caching and offloading content in the multi-community multi-hop D2D networks is modeled by a simple graph, the effect of the communities in offloading social content on each other will not be fully characterized. Thus, we must consider the graph representation of different layers: the physical D2D layer and the different layers of the social graphs. Although the interplay among social communities pertaining to a D2D graph is very challenging, we use a hypergraph framework that is a useful mathematical tool to analysis complicated relationships among multiple entities ~\cite{liu2010modeling,Bai2016} and~\cite{roy2015measuring}. The hypergraph model allows capturing, not only the effect of the social ties between members in each community but also the effect of the interaction between different communities on offloading the cached content. Hence, we model the set of individuals belonging to one community using the hyperedge of a hypergraph, and, then, we can apply hypergraphs to model the multi-hop D2D network while taking into account the presence of multiple communities.

\begin{definition}\label{Hypergraph_def}
\textnormal{Let $\mathcal{N}$ be a finite set of the UEs, a \emph{hypergraph} $H=(\mathcal{L}_1,\mathcal{L}_2,...,\mathcal{L}_L)$ is a family of subsets of $\mathcal{N}$ such that $\mathcal{L}_l \neq \emptyset$ and $\cup_{l=1}^{L}\mathcal{L}_l=\mathcal{N}$. The UEs of $\mathcal{N}$ are the vertices of the hypergraph, and the sets of communities $\mathcal{L}_1,\mathcal{L}_2,...,\mathcal{L}_L$ are called hyperedges.}
\end{definition}

Hence, a hypergraph is a generalized graph in which edges can consist of any subset of the vertices while an edge can exactly consist two vertices in the traditional graph~\cite{Bai2016} and~\cite{liu2010modeling}. One way of analyzing a hypergraph is to model it using a line graph and then analyzing the modeled line graph~\cite{liu2010modeling}. The line graph of a hypergraph is a weighted graph whose vertices are the hyperlinks of the hypergraph and the weight on each edge is related to the interaction between two hyperedges of the hypergraph. Following the properties of the multi-community multi-hop D2D network, we define two weighted graphs from the hypergraph model.
\begin{definition}
\textnormal{A \textit{directed influence-weighted graph} of a hypergraph $H$ is defined as a directed weighted graph $D_i(H)=(\{1,...,L\},\mathcal{E}_H,w_i)$, in which each node of $D_i(H)$ represents one of the communities in $\mathfrak{L}$, $\mathcal{E}_H=\{(u,v)| \mathcal{L}_u,\mathcal{L}_v \in \mathfrak{L}, \exists\text{ }m\text{ and }m' \in \mathcal{L}_u : \mathcal{P}^d_{mm'} \cap \mathcal{L}_u \neq {\O} \}$ and $w_i:\mathcal{E}_H\rightarrow R$.}
\end{definition}

$\mathcal{E}_H$ captures the fact that, if the shortest path between two UEs in community $\mathcal{L}_u$ passes through community $\mathcal{L}_v$ on the D2D graph, then community $\mathcal{L}_u$ will affect community $\mathcal{L}_v$. The reason is that the social content of $\mathcal{L}_u$  passes through some UEs in community $\mathcal{L}_v$ and, thus, some UEs in $\mathcal{L}_v$ receive the social content of $\mathcal{L}_u$ over D2D links. The weight of an edge $\{u,v\}\in D_i(H)$ where $\{u,v\}\in \mathcal{E}_H$, is given by:
\begin{equation}
w_i(\{u,v\})=\sum_{\substack{\forall m,m'\in \mathcal{L}_u \\ \mathcal{P}^d_{mm'} \cap \mathcal{L}_v \neq {\O}}} {w_{mm'}}\times |\mathcal{P}^d_{mm'}\cap \mathcal{L}_v|,
\label{weight of Di_community}
\end{equation}
where $\{(m,m')|\forall m,m'\in \mathcal{L}_u, \mathcal{P}^d_{mm'}\cap \mathcal{L}_v \neq {\O}\}$ is the set of pairs in community $\mathcal{L}_u$ that the shortest path between them passes through community $\mathcal{L}_v$, $w_{mm'}$ is the social tie between UE $m$ and $m'$ belonging to community $\mathcal{L}_u$, and $|\mathcal{P}^d_{mm'}|$ is the number of UEs along the shorest path between UE $m$ and $m'$. Thus, the directed community-weighted graph $D_i(H)$ captures the social tie of the end pair of each path and also the path length that one community provides for other community on the D2D graph.
\begin{definition}
\textnormal{A \textit{directed connectivity-weighted graph} of hypergraph $H$ is defined as a directed weighted graph $D_c(H)=(\{1,...,L\},\mathcal{E}_H,w_c)$, where vertices refer to the communities in $\mathfrak{L}$,  $\mathcal{E}_H=\{(u,v)| \mathcal{L}_u,\mathcal{L}_v \in \mathfrak{L}, \exists\text{ }m\text{ and }m' \in \mathcal{L}_u : \mathcal{P}^d_{mm'} \cap \mathcal{L}_u \neq {\O} \}$, and $w_c(\{u,v\})=|\{(m,m')|\forall m,m'\in \mathcal{L}_u, \mathcal{P}^d_{mm'}\cap \mathcal{L}_v \neq {\O}\}|$.}
\end{definition}
The weight on the link from community $\mathcal{L}_u$ to the community $\mathcal{L}_v$ is equal to the number of pairs in community $\mathcal{L}_u$ that the shortest path between them passes through $\mathcal{L}_v$. Thus, the directed connectivity-weighted graph captures the effects of two communities on each others when the UEs of these two communities provide shortest path for each other in the D2D graph.
\vspace{-0.35cm}
\subsection{Proposed Approach for Content Placement}
To solve the influence optimization problem in (\ref{opt_problem}) within a multi-community, multi-hop D2D network, we must consider two parameters: a) the effect of each member on other members in each community and b) the effect of one community on other communities. If we use the game model presented in (\ref{value function}) for the social graph of each community, then the SV of each player will capture only the exclusive influence of each user on other members of its community. If we use the game in (\ref{value function}) for the directed community-weighted and connectivity-weighted line graphs of the hypergraph $H$, then, the Shapley value of each player will capture the exclusive influence of one community on other communities. To capture the influence of each UE on its community's members and the members of other communities using a single metric, we define an offloading power metric for each UE. A larger offloading power means that the UE has a higher capability for offloading social content in multi-community multi-hop D2D graph. The offloading power of a given UE $k$ that belongs to a community $\mathcal{L}_j$ on the line graph $D$ of the hypergraph $H$ is given by:
\begin{equation}
O_k=\frac{\phi_k(G_j^s)}{\sum_{m\in \mathcal{L}_j} \phi_m(G_j^s)}\phi_j(D),
\label{shaply_value_CF_H}
\end{equation}
where $\phi_k(G_j^s)$ is the Shapley value of UE $k$ on its social graph $G_j^s$ as given by (\ref{shaply_value_CF}), and $\phi_k(D)$ is the Shapley value of community $\mathcal{L}_j$ over the directed weighted line graphs $D_i(H)$ or $D_c(H)$ of hypergraph $H$. $\phi_k(D)$ is given by (\ref{shaply_value_CF}). After defining $O_k$, we compute the offloading power of each UE in D2D graph using equation (\ref{shaply_value_CF_H}). In each community, the UE that has the highest offloading power among the members of its community, is selected as the seed from that community. Thus, the seed set includes the $L$ UEs that have the highest offloading powers among the members of its community.
%*****************************************************************
\vspace{-0.4cm}
{\section{Complexity Analysis}
\label{Dec:Complexity}
The complexity of proposed approach stems from the computation of the Shapley value. A direct application of the original Shapley value formula involves considering $O(2^{|\mathcal{N}|})$ coalitions~\cite{young2014handbook}. Such an exponential complexity in the number of users can be prohibitive for bigger networks. However, based on see Theorem 1, the complexity of calculating the exact formula for the Shapley value in (\ref{shaply_value_CF}) is restricted to the user degree and the shortest path between two users in influence graph. Since, the complexity of calculating the user degree is $O(|\mathcal{N}|)$ and the complexity of calculating the shortest path between two users, is $O(|\mathcal{E}_d|+|\mathcal{N}|\log |\mathcal{N}|)$~\cite{aadithya2010efficient}. Consequently the complexity of the proposed approach is $O(|\mathcal{N}||\mathcal{E}_d|+|\mathcal{N}|^2\log |\mathcal{N}|)$, which is reasonable for the type of problems we are dealing with.}

In practice, the worst case situation for computing the calculated exact formula for the Shapley value in (\ref{shaply_value_CF}) is not likely to occur. This is due to the fact that, in practical scenarios, the probability that every user is reachable from all other users that are within a cutoff distance is low. {For example the proposed approach can be used in real scenarios in the context of a D2D local area network (LAN)~\cite{Last}. In D2D LAN scenarios, the users have a cluster-based distribution such as in a campus, coffee shop, mall, or football stadium, and in each cluster, the number of users as well as the diameter of D2D graph are relatively small.} Thus, the complexity of proposed approach in such practical and cluster-based scenarios is acceptable.
\vspace{-0.4cm}
\section{Simulation Results}
\label{Dec:Simulation}

For simulations, we compare the analytical results of the proposed social-aware framework with other conventional centrality approaches. We consider three metrics indicating the seed set: SV metric on the influence graph (SV) (\ref{shaply_value_CF}), offloading power from the hypergraph modeled by the directed influence-weighted graph (SV:influence) (\ref{shaply_value_CF_H}), and  offloading power from hypergraph modeled by the directed connectivity-weighted graph  (SV:connectivity) according to (\ref{shaply_value_CF_H}). The baselines used for comparison are the conventional centrality measures: degree centrality, betweenness centrality, and closeness centrality~\cite{zhang2014recent} and~\cite{aadithya2010efficient}. We consider a BS at the center of a circular area having a radius of $1$~km. We consider that the UEs form spatial clusters in this circular area. The locations of cluster centers are a realization of a Poisson point process and the UEs are randomly distributed around the cluster centers' locations. The UEs are randomly associated to the communities. The strength of the social tie between any two UEs in one community is uniformly selected from $0$ to $1$. We consider $1$ millisecond for each time slot. The bandwidth of each RB is $15$ kHz. We consider $2$ GHz as a carrier frequency. The maximum power of each UE is $10$~mW which can be equally divided among RBs. The noise power spectral density $N_0$ is considered to be $-170$ dBm per Hz. We assume a path loss exponent $2.5$ and a Rayleigh fading with mean $1$ for the channel model of the D2D links. We set the length of each packet to $100$ bits, and the target bit error rate to $10^{-7}$. All statistical results are averaged over a large number of independent runs.

Fig.~\ref{distance} shows the impact of the distance parameter $d$ on the offloading speed of the cached social content. The offloading speed of social content is the average difference between the number of UEs that received the cached data with D2D sharing during two consecutive time slots. From Fig.~\ref{distance}, we can see that the offloading speed increases suddenly when the distance parameter increases, and then it decreases for high distance. Fig.~\ref{distance} shows that the offloading speed reaches a maximum value when the distance parameter is around $30\%$, $40\%$ and $50\%$ of network diameter for SV, SV:connectivity, and SV:influence approaches, respectively. Clearly, the SV:influence approach has the highest offloading speed and the SV scheme achieves the lowest speed. In Fig.~\ref{distance}, for a low value of $d$, the $d$-distance neighbor set of each UE will not be far from this UE. In this case, the probability that the one-hop neighbor set of each UE will have common members with the $d$-distance neighbor set of other UEs decreases. Thus, the size of one-hop neighbor set becomes more effective in increasing the SV following (\ref{shaply_value_CF}). Consequently, UEs having a large number of one-hop neighbors which are in a crowded portion of the graph are selected as seeds for low $d$. However, for a larger $d$, the $d$-distance neighbors of each UE will be located at a relatively far location from this UE. Thus, the one-hop neighbors of one UE can be in $d$-distance of other UEs with more probability. Consequently, the UEs which are in sparse part of the graph are selected as a seed for high $d$. However, when the distance $d$ is around $40\%$ of the network diameter, neither the UE in sparse nor the UEs in crowded area of D2D graph are selected as seeds. In this case, the UEs in sparse and crowded areas receive cached content with lower hops from the selected seeds. Consequently, the average offloading speed is maximized for the $d$ around $40\%$ of network diameter. The offloading speed resulting from the SV:connectivity and SV:influence approaches is higher than SV approach, because they consider the effect of social tie and social community in selecting seeds while the SV just apply social tie information among UEs. Finally, the offloading speed of the SV:connectivity approach is lower than that of the SV:influence since this latter considers not only the number of connection between two social communities on D2D graph but also the social tie of these connections.
\begin{figure}[t]
\centering
\includegraphics[width=2.75 in]{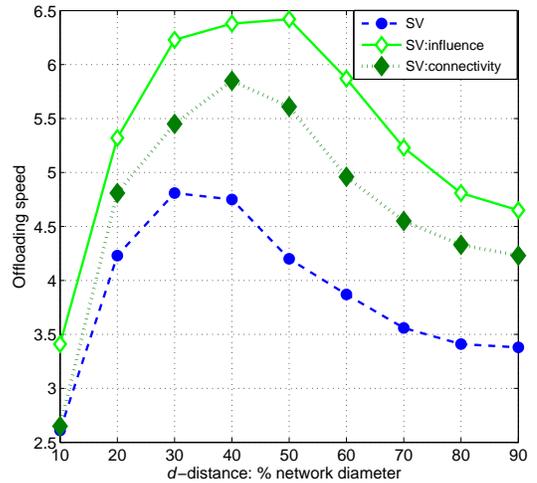}
\caption{\small Average offloading speed vs. distance $d$, when the number
of clusters is 10 and the average number of UEs per cluster is 10.}
\vspace{-0.9cm}
\label{distance}
\end{figure}

Fig.~\ref{hist} shows the distribution of the UEs' SV and offloading power which are normalized by their maximum value. The number of UEs that have the highest normalized value are 2, 6, and 11 for the SV, SV:connectivty and SV:influenced approaches, respectively. In Fig.~\ref{hist}, we can see that the normalized SV of the most UEs is around $0.4$ , while the SV:connectivity and SV:influence approaches assign $0.8$ normalized offloading power to most of the UEs. Moreover, Fig.~\ref{hist} shows that, the number of UEs that have a high normalized offloading power is larger than those having high normalized SV. This means that there are some UEs that are considered effective under SV:connectivity and SV:influence approaches while the SV approach categorized them as ineffective in offloading social content over the D2D graph. This is due to the fact that UEs that have a low SV will have a low connectivity degree and UEs in their one-hop neighbor set are in the $d$-distance neighbor set of other effective UEs. Thus, the number of UEs with high SV in one community is low. These UEs with low SV in one community can have more offloading power, because they may be in the shortest path set between the members of other community. Thus, they can be critical in offloading the content between the members of other communities.
\begin{figure}[t]
\centering
\includegraphics[width=2.75 in]{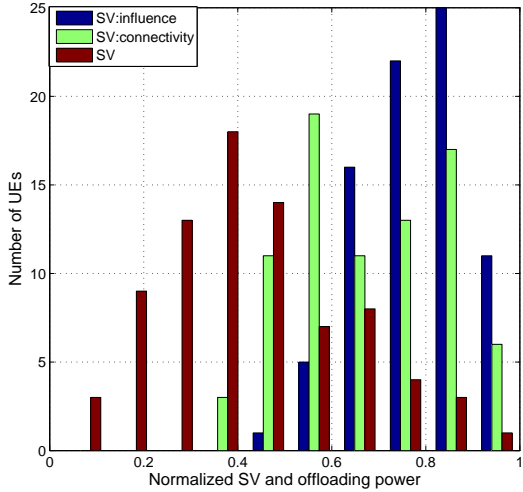}
\caption{\small The histogram of SV and offloading power of UEs when UEs are distributed according clutter point process.}
\vspace{-0.7cm}
\label{hist}
\end{figure}

Fig.~\ref{Naffected} shows the impact of the number of UEs on the number of influenced UEs within the D2D graph when $d=40\%$ of the network diameter. From Fig.~\ref{Naffected}, we can see that the number of influenced UEs increases when the number of UEs increases. Clearly, the number of influenced UEs resulting from the proposed schemes is higher than the number of influenced UEs resulting from conventional centrality approaches. The betweenness centrality behavior is closest one to the Shapley value centrality. The three conventional centrality approaches do not consider the common effect of other seeds on sharing cached content with other UEs. The closeness and degree centralities usually choose the seeds which are close to each other in the crowded parts of the D2D graph, and betweenness centrality usually choose the seeds which are in the most of the shortest path sets between other UEs. Thus, the UEs located in the crowded portions of the D2D graph which are connected to other crowded area are selected as the seeds. The proposed approaches increase the exclusive effect of each seed by decreasing the common members among the $d$-distance neighbors of the one-hop neighbors of the selected seeds. Thus, the seeds resulting from the proposed approaches are distributed across the D2D graph. Since the SV:connectivity and SV:influence approaches consider the the effect of social communities on each others, the seed set selected by these approaches has the most number of influenced UEs. {On the average, the SV:influence achieves $12\%$, $19\%$ and $21\%$ improvement in the number of affected UEs compared to betweenness, degree, and closeness approaches, respectively.}
\begin{figure}[t]
\centering
\includegraphics[width=2.75in]{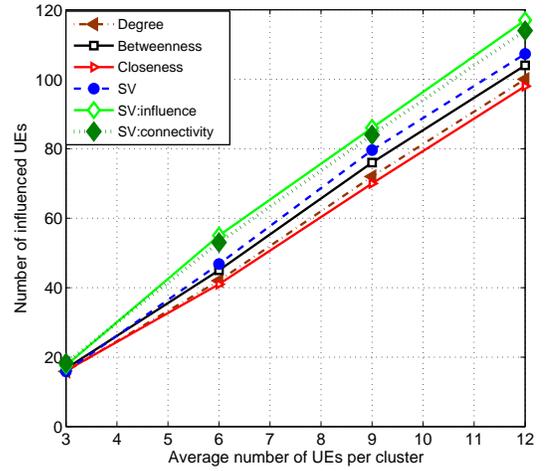}
\caption{\small Number of affected UEs for $d_n=40\%$ of D2D graph diameter when the number
of cluster is 10.}
\vspace{-0.7cm}
\label{Naffected}
\end{figure}

Fig.~\ref{NSpeed} shows the impact of the number of UEs on the average social content offload speed when $d=40\%$ of network diameter. From Fig.~\ref{NSpeed}, we can see that the offload speed increases when the number of UEs increases. From this figure, we can see that degree centrality, our proposed approaches, and  betweenness centrality have the maximum, medium, and minimum average speed of the offloading social content, respectively. The reason is that, in the average degree centrality, the seeds offload social content to many UEs in their one-hop neighbors. The average speed of degree centrality decreases in the last time slots because the total number of affected UEs is low (Fig.~\ref{Naffected}). {Although the speed of offloading of the degree centrality is the highest, the total number of UEs affected in degree centrality is near the lowest one (Fig.~\ref{Naffected}). Hence, in degree centrality, the seeds can quickly share their cached content. However, the number of UEs affected by more than one seed is large because the degree centrality does not consider the exclusive influence of seeds. In contrast, for the proposed approaches, even though the offloading speed is below the degree centrality, the total number of affected UEs is the highest (see Fig.~\ref{Naffected}). The reason is that the number of UEs affected only by one seed is high, since the exclusive influence of each seed is high in our proposed approach. Thus, it takes more time to share the cached content to other UEs. However, the content is now received by a large number of UEs.}
\begin{figure}[t]
\centering
\includegraphics[width=2.75in]{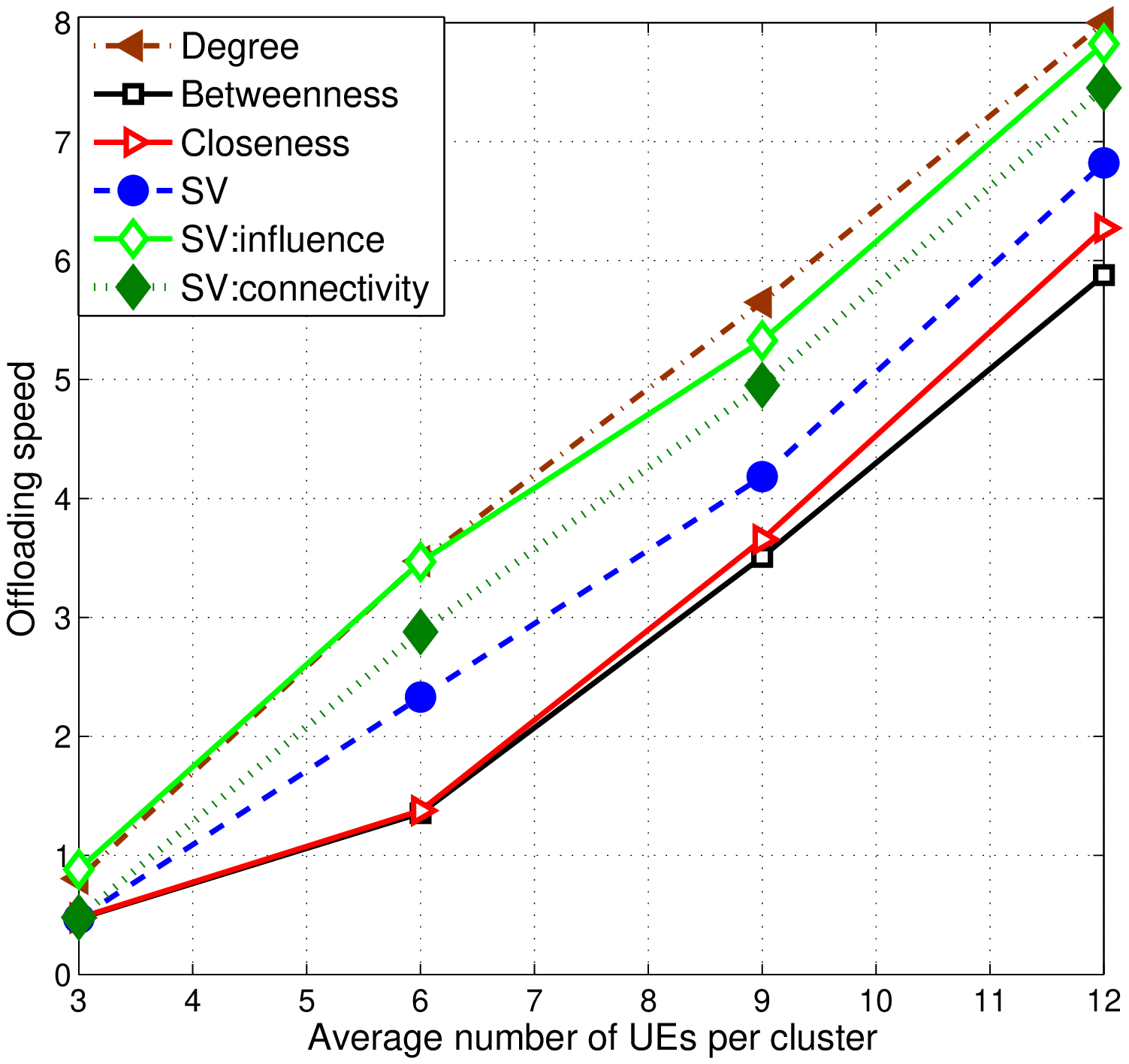}
\caption{\small The offloading Speed of the social cached content for $d_n=40\%$ of D2D graph diameter when the number of cluster is 10.}
\vspace{-0.8cm}
\label{NSpeed}
\end{figure}

For a real-world case, we use Netvizz, a data collection and extraction application that allows exportation of data in standard file formats from different sections of the Facebook social networking service~\cite{rieder2013studying}. Friendship networks, groups, and pages can thus be analyzed quantitatively and qualitatively with regards to demographical, postdemographical, and relational characteristics~\cite{rieder2013studying}. We extract the data from three group pages of students at Virginia Tech on Facebook. These groups are related to sports clubs that gather communities of students who are interested in a common sport. Based on the results from Netvizz, the social network of each community is extracted in which the social tie among two nodes captures the social interest of two members in common data such as a certain sport video file. For the D2D graph capturing the locations of users and their possible D2D links, we distribute the users based on a cluster process over an area of $1000$m$\times1000$m that represents a campus area, then we randomly assign each user to one of the members of three communities. The strength of the social tie between any two UEs in one community is based on the data extracted from the Facebook group pages.

Fig.~\ref{Naffected_real} shows the impact of the number of UEs on the number of influenced UEs within the D2D graph for a real-world case. Fig.~\ref{Naffected_real} shows that the number of influenced UEs increases with respect to the number of UEs. Moreover, the proposed schemes yield a higher number of influenced UEs compared to the conventional centrality measures. The reason is that the proposed approaches increase the exclusive effect of each seed over D2D graph. In addition, the SV:connectivity and SV:influence approaches consider the mutual effect of social communities. {Hence, the seed sets selected by these approaches yield the largest number of influenced UEs. For the real-world case, on the average, the SV:influence achieves $13\%$, $14\%$, and $16\%$ improvement in the number of affected UEs compared to betweenness, degree, and closeness approaches, respectively.}
\begin{figure}[t]
\centering
\includegraphics[width=2.75in]{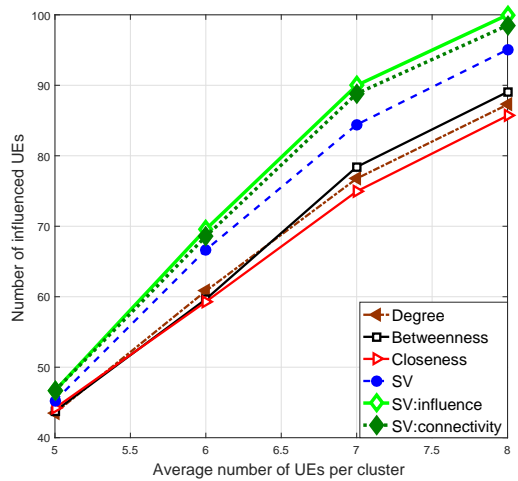}
\caption{\small Number of affected UEs  for a real-world case.}
\vspace{-0.5cm}
\label{Naffected_real}
\end{figure}

Fig.~\ref{NSpeed_real} shows the impact of the number of UEs on the average social content offload speed for a real-world case. Fig.~\ref{NSpeed_real} shows that the offload speed increases with respect to the number of UEs. From this figure, we can see that degree centrality, our proposed approaches, and  betweenness centrality achieve the maximum, medium, and minimum average speed of the offloading social content, respectively. The reason is that, the speed of offloading is high at the initial time slot for the average degree centrality because the seeds offload social content to many UEs in their one-hop neighbors. {Although the speed of offloading of degree centrality is the highest, the total number of UEs affected in degree centrality is nearly the lowest one as shown in Fig.~\ref{Naffected_real}. This means that the seeds can quickly share their cached content, but they share it with common UEs. Although the speed of offloading of the proposed approaches is lower than degree centrality, the proposed approaches outperform all other centrality measures.}
\begin{figure}[t]
\centering
\includegraphics[width=2.75in]{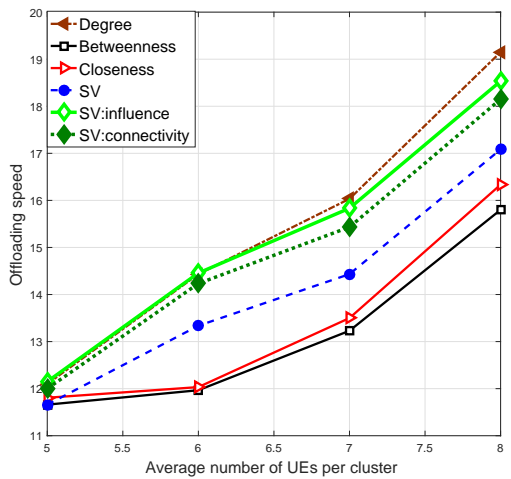}
\caption{\small The offloading Speed of the social cached content for a real-world case.}
\vspace{-0.7cm}
\label{NSpeed_real}
\end{figure}

\begin{figure}[t]
\centering
\includegraphics[width=2.75in]{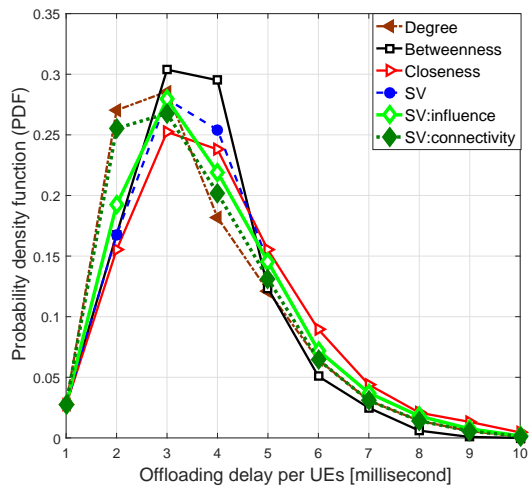}
\caption{\small {The probability density function of offloading delay per UE.}}
\vspace{-0.3cm}
\label{PDF_delay}
\end{figure}
{Fig.~\ref{PDF_delay} shows the probability distribution function (PDF) for the offloading delay per UE. As we can see from Fig.~\ref{PDF_delay}, the average offloading delay per UE for the SV, SV:connectivity, and SV:influence approaches are respectively 3.86, 3.77, and 3.58 milliseconds, and for the degree centrality the stochastic average offloading delay is 3.54 milliseconds. {Thus, on the average, the average offloading delay per UE resulting from the proposed SV-based approaches is only about 0.2 milliseconds higher than degree centrality.}}

{To further illustrate how our approach can work for mobile cases, we simulate a new setup in which the users move. In our simulation, the social graph of each community is stable while the D2D graphs randomly change over time. We move 170 users with different speeds based on a random walk model. We consider a higher correlation between the movement patterns of users that have strong social tie. After each change in the D2D graph due to the users' mobility, we compare two scenarios: using the initially selected seeds for the new, modified D2D graph (A); and selecting new seeds based on the new D2D graph (B). Fig.~\ref{Dif_number} shows the difference between the number of influenced UEs for the two seed selection scenarios A and B as function of the users' speed. As we can see from Fig.~\ref{Dif_number}, the difference between the number of influenced UEs increases with the speed of the UEs. This is due to the fact that, by increasing the speed of the UEs, the dynamic change in the D2D graph increases and the seed selection in scenario B has to recompute the selected seed while scenario A retains the initially selected seeds. On the average, the difference between the number of influenced UEs for these two scenarios A and B is around 4 for SV-based approaches and less than 6 for other approaches, which represents a very small fraction of users. This stems from the correlated mobility patterns of the users~\cite{R2_51, R2_52}.  Thus, for low-speed mobile users, if we do not recompute the SV after every change in D2D network, just 4 out of 170 users (around 2\%) will not be influenced which demonstrates the effectiveness of our approach even for mobile cases.}
\begin{figure}[t]
\centering
\includegraphics[width=2.75in]{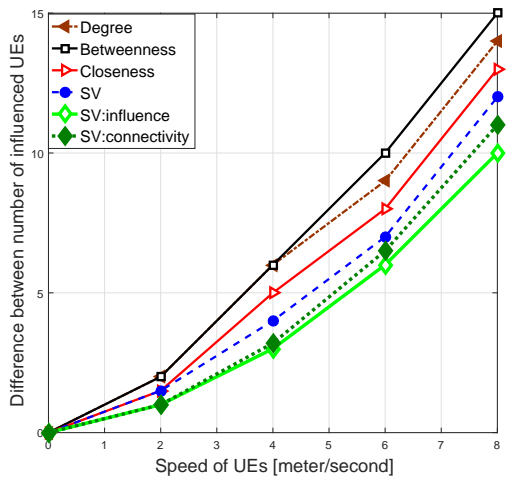}
\caption{\small {Difference between number of influenced UEs.}}
\vspace{-0.6cm}
\label{Dif_number}
\end{figure}
%*****************************************************************
\vspace{-0.4cm}
\section{Conclusion}
\label{Dec:Conclusion}
In this paper, we have proposed a novel context-aware framework for cache placement at wireless UEs in order to improve social content offloading in a D2D-enhanced cellular network. In this network, the users belong to different social communities while UEs form a single multi-hop D2D network. We exploit the multi-community social context of users for improving  the local offloading of cached content by allowing an effective use of multi-hop D2D sharing. Based on the social tie of the users, a cooperative game between UEs is proposed. The value of a coalition is equal to the $d_n$-influence of its members on other UEs over the D2D graph. We have proved that Shapley value of each UE in the proposed cooperative game shows the exclusive effect of UE in content offloading over D2D links. Due to social tie between members of each community and D2D links between UEs, we have modeled the cache placement problem using hypergraph that is analyzed using two line graph models. Using the proposed line graphs coupled with the SV derived from the cooperative game, we have defined an offloading power for each UE in multi-community multi-hop D2D network. Hence, we have considered the UEs with high offloading power which have more exclusive effect on both of its and other community's members as the cache placements. Simulation results have shown that on the average the proposed approach yields significant improvements in terms of the number of UEs that offload popular content, compared to the schemes based on the classical centrality measures.

%*****************************************************************
\bibliographystyle{IEEEtran}
\bibliography{references}
%\vspace{0.2cm}

\begin{IEEEbiography}[{\includegraphics[width=1in,height=1.25in,clip,keepaspectratio]{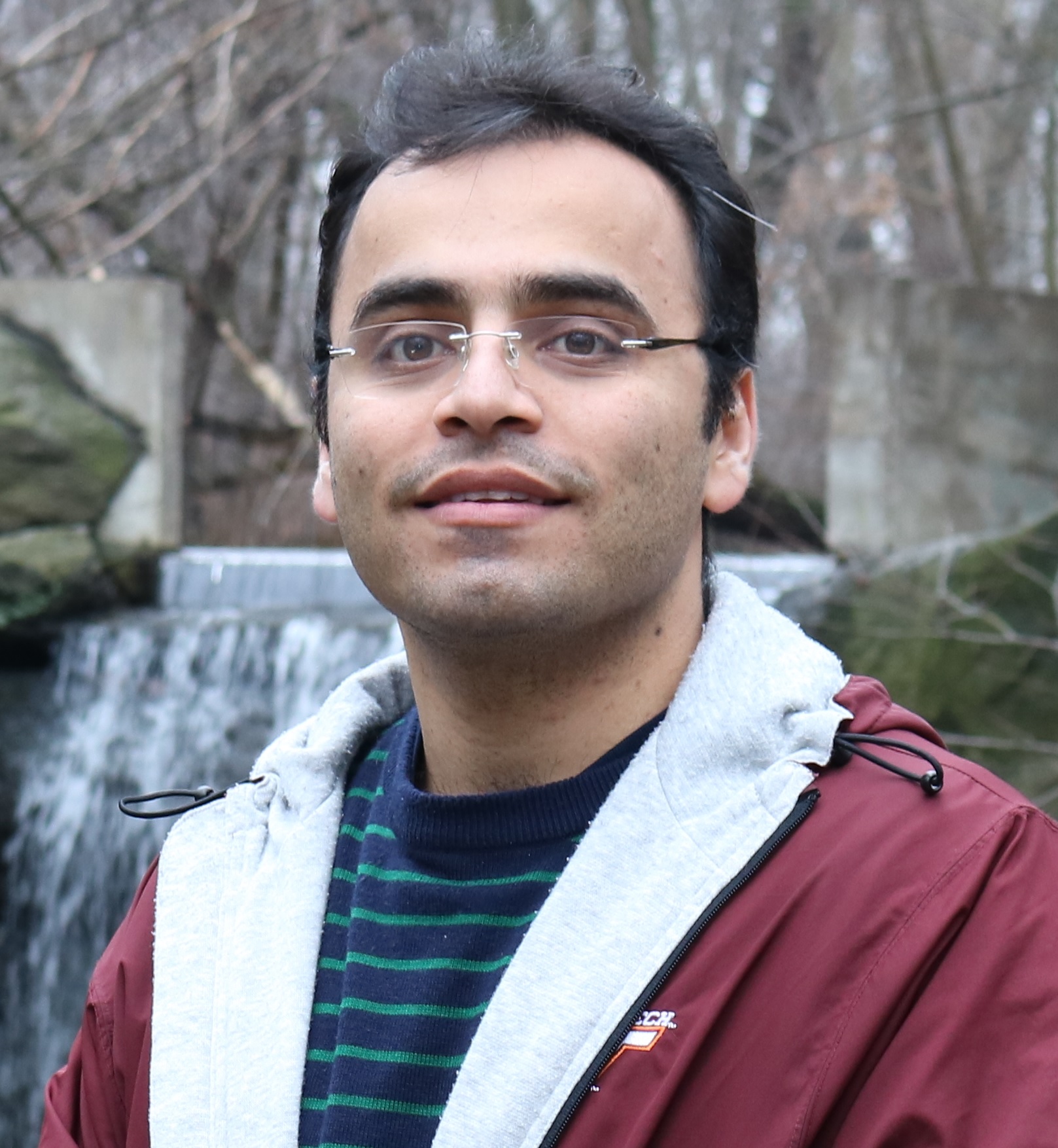}}] {Mehdi~Naderi~Soorki} received the B.Sc degree in Electrical Engineering from Iran University of Science and Technology in 2007. He received his M.Sc and Ph.D degrees in Telecommunication Networks from the Isfahan University of Technology in 2010 and 2018, respectively. He held visiting position at the Wireless@VT group at the Dept. of Electrical and Computer Engineering at Virginia Tech. His research interests include wireless networking, game theory and stochastic optimization.
\end{IEEEbiography}
\vspace{-1.38cm}
\begin{IEEEbiography}[{\includegraphics[width=1in,height=1.25in,clip,keepaspectratio]{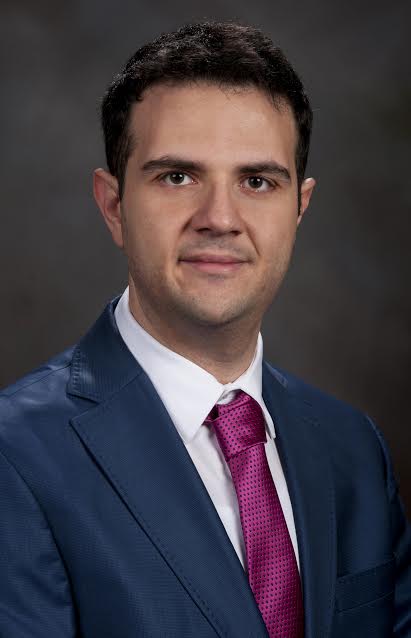}}]
{Walid~Saad}(S'07, M'10, SM’15) received his Ph.D degree from the University of Oslo in 2010. Currently,  he is an Associate Professor at the Department of Electrical and Computer Engineering at Virginia Tech, where he leads the Network Science, Wireless, and Security (NetSciWiS) laboratory, within the Wireless@VT research group. His  research interests include wireless networks, machine learning, game theory, cybersecurity, unmanned aerial vehicles, and cyber-physical systems. Dr. Saad is the recipient of the NSF CAREER award in 2013, the AFOSR summer faculty fellowship in 2014, and the Young Investigator Award from the Office of Naval Research (ONR) in 2015. He was the author/co-author of six conference best paper awards at WiOpt in 2009, ICIMP in 2010, IEEE WCNC in 2012,  IEEE PIMRC in 2015, IEEE SmartGridComm in 2015, and EuCNC in 2017. He is the recipient of the 2015 Fred W. Ellersick Prize from the IEEE Communications Society. From 2015-2017, Dr. Saad was named the Stephen O. Lane Junior Faculty Fellow at Virginia Tech and, in 2017, he was named College of Engineering Faculty Fellow. He currently serves as an editor for the IEEE Transactions on Wireless Communications, IEEE Transactions on Communications, IEEE Transactions on Mobile Computing, and IEEE Transactions on Information Forensics and Security.
\end{IEEEbiography}
\vspace{-1.38cm}
\begin{IEEEbiography}[{\includegraphics[width=1in,height=1.25in,clip,keepaspectratio]{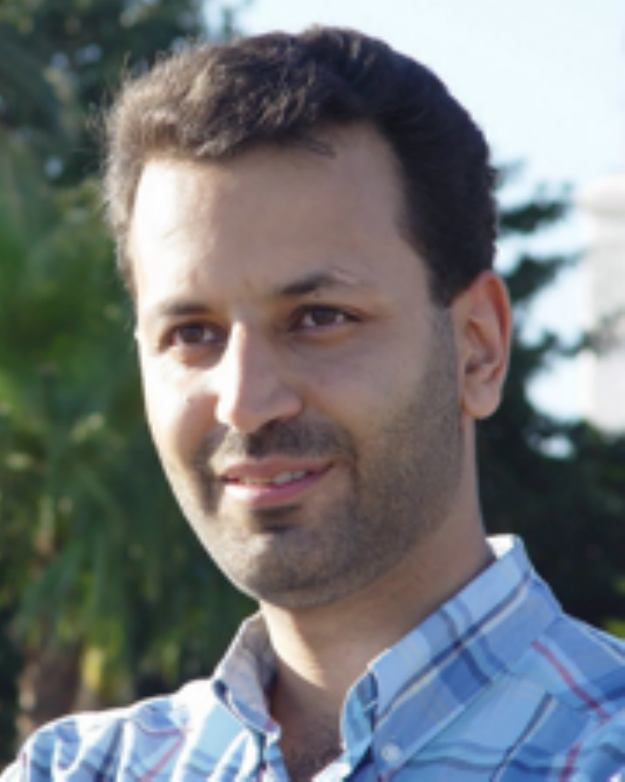}}]
{Mohammad~Hossein~Manshaei} received the BSc degree in electrical engineering and the MSc degree in communication engineering from the Isfahan University of Technology in 1997 and 2000, respectively. He received another MSc degree in computer science and the PhD degree in computer science and distributed systems from the University of Nice Sophia-­Antipolis, France, in 2002 and 2005, respectively. He did his thesis work at INRIA, Sophia-­Antipolis, France. He is currently an Associate Professor at the Isfahan University of Technology, Iran. From 2006 to 2011, he was a senior researcher and lecturer at EPFL, Switzerland. He held visiting positions at the UNCC, the NYU, and the VTech. His research interests include wireless networking, wireless security and privacy, computational biology, and game theory.
\end{IEEEbiography}
\vspace{-1.38cm}
\begin{IEEEbiography}[{\includegraphics[width=1in,height=1.25in,clip,keepaspectratio]{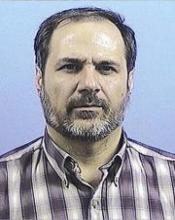}}]
{Hossein~Saidi} received B.S and M.S. degrees in Electrical Eng. in 1986 and 1989 respectively, both from Isfahan University of Technology (IUT), Isfahan Iran. He also received D.Sc. in Electrical Eng. from Washington University in St. Louis, USA in 1994. Since 1995 he has been with the Dept. of Electrical and Computer Engineering at IUT, where he is currently Full Professor and serves as the Chair of of Electrical and Computer Engineering Department. His research interest includes high speed switches and routers, communication networks, QoS in networks, security, queueing system and information theory. He holds 4 USA and one International patents and has published more than 100 scientific papers. He is the recipient of several awards including: 2006 ASPA award (The 1st Asian Science Park Association leaders award) and the Certificate award at 1st National Festival of Information and Communication Technology (ICT 2011) both as the CEO of SarvNet Telecommunication Inc.
\end{IEEEbiography}

%\balance
\end{document}